\documentclass{article}
\usepackage{iclr2027_conference}
\usepackage{times}
\usepackage[T1]{fontenc}
\usepackage[utf8]{inputenc}
\usepackage{amsmath,amssymb,amsthm,mathtools}
\usepackage{booktabs,multirow}
\usepackage{hyperref}
\usepackage{url}
\usepackage[nameinlink,capitalize,noabbrev]{cleveref}

\newcommand{\R}{\mathbb{R}}
\newcommand{\E}{\mathbb{E}}
\newcommand{\Prb}{\mathbb{P}}
\newcommand{\OPT}{\operatorname{OPT}}
\newcommand{\norm}[1]{\left\lVert #1\right\rVert}
\newcommand{\ones}{\mathbf{1}}
\newcommand{\col}{\operatorname{col}}

\newcommand{\argmin}{\operatorname*{arg\,min}}
\newtheorem{theorem}{Theorem}[section]
\newtheorem{lemma}[theorem]{Lemma}

\theoremstyle{definition}

\theoremstyle{remark}

\hypersetup{colorlinks=true,linkcolor=blue,citecolor=blue,urlcolor=blue,
  pdftitle={Tight Efficiency Guarantees for Strategyproof Linear Regression},
  pdfauthor={Anonymous authors},pdfsubject={ICLR 2027 submission}}

\iclrfinalcopy
\title{Tight Efficiency Guarantees for\\Strategyproof Linear Regression}
\author{Yichen Huang\thanks{Equal contribution; listed in alphabetical order.},\, Yuqi Pan\footnotemark[1],\, Michael Mitzenmacher, Milind Tambe, Yiling Chen \\
Harvard University \\
\texttt{\{yichenhuang, yuqipan\}@g.harvard.edu}
}

\begin{document}
\maketitle
\begin{abstract}
We study the trade-off between squared-error accuracy and incentive compatibility in linear regression. Agents report private labels associated with publicly known features and prefer predictions close to their true labels. Ordinary least squares (OLS) need not elicit truthful reports.

For regression with $d$ parameters, we design a deterministic group-strategyproof mechanism achieving a $(d+1)$-approximation to the least-squares optimum and prove optimality even among universally strategyproof randomized mechanisms, answering an open question of \citet{chen2018}. Relaxing universal strategyproofness to strategyproofness in expectation reveals a sharp separation: squared individual loss retains the factor $d+1$, while absolute individual loss admits the tight ratio $2-1/(\lceil d/2\rceil+1)$. 
\end{abstract}

\section{Introduction}
\label{sec:introduction}

Machine learning usually starts with a dataset and asks how to build an accurate predictor from it. But when a deployed model informs decisions that affect the people supplying its data, the learning rule can also influence what they report~\citep{perdomo2020,kleinebuening2025strategyproof}. In such settings, where people may have misaligned preferences over the predictors, we must ask not only how accurately it predicts, but also whether people benefit from manipulating its inputs. 

One prominent example is linear regression, a widely used model in many deployed prediction systems such as electricity-demand forecasting~\citep{hong2014load,pjm2025load}, residential property valuation~\citep{mpac2016residential}, and medical-cost prediction~\citep{cms2024risk}. When predictions guide resource allocation, the people supplying information may have a reason to distort it. For example, Zara's demand forecasts incorporated store managers' shipment requests, which managers frequently inflated when popular items were scarce \citep{caro2010zara}. 

Rather than taking honest inputs for granted, strategyproof mechanism design asks how a learning rule can make truthful reporting a best response. In \emph{strategyproof linear regression}, features are public, each agent reports one private response, and the output is a linear predictor. Each agent wants the prediction at their features to be close to their true response. A strategyproof rule makes honesty optimal whatever others report, without payments, label verification, a prior over agents' information, or equilibrium coordination. In a laboratory study of regression, \citet{perote2015strategic} find more misreporting under OLS than under a strategyproof resistant line estimator.

Truthfulness alone does not ensure a useful predictor: a rule that ignores every report is strategyproof. We also want low squared error, the standard regression objective \citep{hastie2009,gneiting2011} and a social objective in mechanism design \citep{feldman2013}. Even fitting a constant without public features illustrates the tension. Least squares returns the mean, which an agent can move toward their true value by misreporting. A suitably tie-broken median prevents manipulation but can increase squared error. The approximation ratio against the least-squares fit to the true labels therefore measures the \emph{price of strategyproofness}.

For linear regression over general feature spaces, this price has remained poorly understood. Prior work establishes truthful rules, including absolute-loss minimization and generalized resistant hyperplane mechanisms \citep{dekel2010,chen2018}, but leaves a substantial gap in their squared-error guarantees: the known deterministic upper bound is $n$, the number of agents, while the general lower bound is only $2$ \citep{chen2018}. It is natural to ask for the best accuracy guarantee achievable under incentive compatibility.

Randomization offers a potential way to further reduce this price. With a carefully designed randomized mechanism, a beneficial deviation under some realizations may hurt the agent under other realizations, leaving truthful reporting optimal in expectation. This flexibility has yielded better guarantees in settings such as facility location~\citep{feldman2013}, but its role in strategyproof linear regression has not been studied. Balancing accuracy and incentives is more delicate here: truthfulness in expectation depends on the precise form of agents' losses, not just their preference for smaller individual errors. We therefore ask whether randomization can lower the price of strategyproofness, and how the answer depends on agents' preferences.

\paragraph{Our results.}
We resolve these questions and show that the relevant price is governed by the dimension. Throughout the paper, \emph{$d$ is the number of free coefficients}, including an intercept if one is fitted. Our social objective is always squared error; an agent's individual loss is the $p$-th power of their absolute error, for some $p>0$ common to every agent. \Cref{tab:results} summarizes the results.

First, we give a deterministic $(d+1)$-approximation (\cref{thm:deterministic}) by minimizing a weighted sum of absolute residuals. The weights depend only on public features and capture how strongly the model can move each prediction. This preserves truthfulness while replacing sample-size dependence by dimension dependence. For closed convex hypothesis sets, the guarantee becomes $k+1$, where $k \le d$ is the affine dimension of the hypothesis set. This approximation ratio is tight, even among universally strategyproof randomized mechanisms.\footnote{A randomized mechanism is universally strategyproof (USP) if it is a distribution over deterministic strategyproof mechanisms, and strategyproof in expectation (SPIE) if no agent can reduce their \emph{expected} loss by misreporting. See \cref{sec:model}.}

Next, we consider strategyproof-in-expectation (SPIE) mechanisms for the complete hypothesis set. Unlike deterministic truthfulness, truthfulness in expectation depends on agents' loss. The most natural choices are arguably absolute error ($p=1$) and squared error ($p=2$).

For $p=1$, randomization gives a substantially better guarantee (\cref{thm:absolute}). We sample a geometrically diverse set of agents, fit their reports exactly, and minimize squared error over all observations using the remaining freedom. Selecting about half as many agents as there are coefficients balances two effects: a selected agent cannot improve on an exact fit, while an unselected agent has limited influence. The resulting ratio is $2-1/(\lceil d/2\rceil+1)$, and a matching lower bound shows its optimality.

In contrast, for $p=2$, individual and social losses are both squared error, and the factor $d+1$ cannot be improved (\cref{thm:lower}). Randomization therefore offers no improvement in the optimal worst-case guarantee.

Finally, the benefit of randomization is not confined to absolute loss. For every positive exponent $p\ne2$, we give a finite lottery over combinations of independently sampled interpolating models that strictly beats $d+1$ (\cref{thm:powers}). Averaging improves accuracy, while additional outcomes protect truthful reporting.

Our results give a relatively complete picture of the accuracy--incentive trade-off in linear regression. The underlying concern extends across machine learning whenever the people supplying data or feedback have a stake in the model they help train~\citep{ouyang2022,chen2023collaborative,wang2024rlhfpoison,kleinebuening2025strategyproof}. We hope this work inspires broader research on trade-offs between statistical accuracy and incentive compatibility in learning systems.

\begin{table}[t]
\centering
\caption{Worst-case guarantees with $d$ coefficients. $p$ refers to the power in the individual loss. Deterministic and USP guarantees do not depend on $p$ in the individual loss function. 
The SPIE upper bounds for $p\ne2$ assume the complete hypothesis set. 
All but the last row are tight.\vspace{1em}}
\label{tab:results}
\begin{tabular}{@{}lll@{}}
\toprule
Incentive requirement & Individual loss & Approximation ratio \\
\midrule
Deterministic / universally SP & -- & $d+1$  \\
\multirow{3}{*}[\dimexpr\normalbaselineskip\relax]{SP in expectation}& $p=1$ & $2-1/(\lceil d/2\rceil+1)$  \\
 & $p=2$ & $d+1$  \\
 & $p>0,\ p\neq 2$ & $<d+1$ \\
\bottomrule
\end{tabular}
\end{table}


\subsection{Related work}
\label{sec:related}

\paragraph{Strategic regression and learning.}
\citet{perote2004} study strategyproof estimators for simple regression. \citet{dekel2010} develop a learning framework for strategic data and establish truthfulness of absolute-loss minimization in the one-point-per-agent setting. \citet{chen2018} extend this approach to weighted objectives and convex regularization, develop high-dimensional mechanisms and characterizations, and pose the squared-error efficiency question. In a complementary direction, \citet{hossain2021} analyze equilibria of non-strategyproof regression rules. Other work considers strategic behavior in classification models~\citep{meir2012,hardt2016} or learning with external predictions \citep{balkanski2025}. Payment-based models address different incentives, including privacy costs \citep{cummings2015,zhu2024truthful} and costly data acquisition \citep{cai2015}.

\paragraph{Facility location and estimation.}
Fitting a constant without public features is equivalent to single-facility location on the line.\footnote{The true value of an individual corresponds to their location and the fitted value corresponds to the selected location of the facility.} Generalized medians connect its incentive properties to classical social choice \citep{moulin1980}, and approximate mechanism design without money studies the accompanying optimization tradeoffs \citep{procaccia2013}. For the sum of squared distances and expected-distance preferences, \citet{feldman2013} give an optimal randomized ratio of $3/2$ on the line, compared with the deterministic optimum of $2$.

\paragraph{Accuracy under additional requirements.}
Beyond incentive compatibility, researchers study trade-offs between predictive accuracy and other desirable properties, such as fairness \citep{menon2018fairness}, differential privacy \citep{chaudhuri2011differentially}, and adversarial robustness \citep{tsipras2019robustness}. Related performance trade-offs arise for interpretable clustering \citep{dasgupta2020explainable} and explainable information design \citep{chen2026price}.


\section{Model and preliminaries}
\label{sec:model}

\paragraph{Data and dimension.}
There are $n$ agents and $d\ge1$ coefficients. Agent $i$ has a public row vector $x_i^\top$ and a private label $y_i\in\R$, and controls only the report of that label. Let $X\in\R^{n\times d}$ have rows $x_i^\top$, and let $y\in\R^n$ be the label vector. We assume $X$ has full column rank, so $n\ge d$. A coefficient vector $\beta\in\Theta\subseteq\R^d$ produces the predictions $X\beta$, where $\Theta$ is a public nonempty closed convex set fixed independently of reports. Unless a constraint is specified, $\Theta=\R^d$.

Our dimension convention counts \emph{free scalar coefficients}, not the dimension of an augmented data point. For example, $f(u)=a^\top u+b$ on $u\in\R^{d-1}$ is a $d$-parameter model, represented by $x=(u,1)$ and $\beta=(a,b)$. 
We use $e_i$ for the $i$-th standard basis vector and $\ones$ for the all-ones vector.

\paragraph{Social objective.}
For $z=X\beta$, write $C(z;y)=\norm{z-y}_2^2$ and $\OPT_\Theta(y)=\min_{\beta\in\Theta}C(X\beta;y)$. A mechanism maps reports to feasible coefficients, possibly at random; $Z(y)$ denotes its random prediction vector. It is \emph{$c$-efficient}, or a \emph{$c$-approximation}, if
\begin{equation*}
\E C(Z(y);y)\le c\,\OPT_\Theta(y)\qquad\text{for every }y\in\R^n.
\label{eq:approximation}
\end{equation*}
The inequality also applies when the optimum is zero: any finite approximation must then return the exact prediction vector almost surely. 

\paragraph{Individual incentives.} Each agent's preference over the mechanism's prediction for them is single-peaked at their true value. Here, each agent evaluates a prediction using $\ell_p(z_i,y_i)=|z_i-y_i|^p$, for a fixed finite exponent $p>0$ shared by all agents. The cases $p=1$ and $p=2$ are absolute and squared individual loss, respectively. A randomized mechanism is \emph{strategyproof in expectation} (SPIE) for loss $\ell_p$ if, for every agent $i$, every label vector $y$, and every alternative report $t\in\R$,
\begin{equation*}
\E|Z_i(y)-y_i|^p\le\E|Z_i(t,y_{-i})-y_i|^p,
\label{eq:spie}
\end{equation*}
where $(t,y_{-i})$ replaces only coordinate $i$ of $y$ by $t$.

A deterministic mechanism is \emph{strategyproof} if the same inequality holds without expectations. 
It is \emph{group-strategyproof} if no coalition can change its reports so that every member weakly reduces their loss and at least one member strictly reduces it.
A randomized mechanism is \emph{universally strategyproof} if it is a report-independent distribution over deterministic strategyproof mechanisms. 

All positive powers of absolute error induce the same preferences over deterministic outcomes. Consequently, deterministic and universal strategyproofness do not depend on $p$.

\section{An optimal deterministic mechanism}
\label{sec:deterministic}

The deterministic mechanism keeps the incentive protection of absolute-loss minimization, but changes the relative weight of observations. A naive $\ell_1$ risk minimizer with equal weights compares an $\ell_1$ objective with an $\ell_2$ objective in the entire $n$-dimensional observation space. Our weights instead recognize that feasible predictions move in at most $d$ independent directions.

\paragraph{Geometry and weights.} We first introduce the notation. Let $K=X\Theta$ be the feasible prediction set. To describe how predictions can change, define its direction space $L=\operatorname{span}(K-K)$, where $K-K=\{z-z':z,z'\in K\}$. Thus $L$ contains all linear combinations of feasible prediction differences and ignores any fixed offset. Set $k=\dim L$; because $X$ has full column rank, $k$ is also the affine dimension of $\Theta$.

Let $P$ be the orthogonal projector onto $L$: for every $v\in\R^n$, $Pv$ is the unique vector in $L$ closest to $v$ in Euclidean distance. To compute it explicitly when $k>0$, choose $A\in\R^{d\times k}$ whose columns form a basis of $\operatorname{span}(\Theta-\Theta)$. Then $L=\col(XA)$ and
\[
P=XA(A^\top X^\top XA)^{-1}A^\top X^\top.
\]
The diagonal entries $\tau_i=P_{ii}$ are the \emph{leverage scores} of $L$; they will determine how much weight each observation receives.

\paragraph{Mechanism: leverage-weighted $\ell_1$-ERM.}
Choose
\begin{equation}
\widehat\beta(y)\in\argmin_{\beta\in\Theta}
\sum_{i=1}^n\sqrt{\tau_i}\,|y_i-x_i^\top\beta|,
\quad\text{breaking ties by minimizing }\norm{\beta}_2^2.
\label{eq:weighted-lad}
\end{equation}

This is a weighted $\ell_1$ risk minimizer, with weights determined by the leverage scores. Intuitively, an observation should receive more weight when feasible models can change its prediction more strongly. Precisely, for $u\in L$, Cauchy--Schwarz gives $|u_i|\le\sqrt{\tau_i}\norm{u}_2$. This motivates the square root in \eqref{eq:weighted-lad}. The leverage scores also satisfy $\sum_i\tau_i=k$, so their total contribution is controlled by dimension. As an example, for constant fitting, every leverage score is $1/n$; the mechanism reduces to a median, with the specified fixed tie-break, and the bound becomes $2$.


\begin{theorem}[Leverage-weighted absolute deviations]
\label{thm:deterministic}
The mechanism in \eqref{eq:weighted-lad} is well defined and group-strategyproof. Its squared-error approximation ratio is at most $k+1$, where $k \le d$ is the affine dimension of $\Theta$.
\end{theorem}

\begin{proof}
Absolute-loss minimization with public nonnegative, report-independent weights and a fixed strictly convex tie-break is group-strategyproof over a convex prediction class \citep[Theorem 1]{chen2018}. It remains to prove the approximation guarantee.

Let $z^\star$ be the least-squares minimizer in $K$, let $\widehat z=X\widehat\beta$, and set $r^\star=y-z^\star$ and $\widehat r=y-\widehat z$. Since $P\widehat r$ and $(I-P)\widehat r$ are orthogonal, we split $\|\widehat r\|_2^2 = \norm{P\widehat r}_2^2 +\norm{(I-P)\widehat r}_2^2$ and bound them separately. 

Both predictions lie in $K$, so their difference lies in $L$ and $(I-P)\widehat r=(I-P)r^\star$. In words, every feasible model has the same residual component perpendicular to the direction space; only the parallel component can change. Thus, 
  \[
  \norm{(I-P)\widehat r}_2^2 = \norm{(I-P)r^\star}_2^2 \le \norm{r^\star}_2^2.
  \]

It remains to bound $\norm{P\widehat r}_2^2$. Define the weighted absolute residual $W(a)=\sum_{i=1}^n\sqrt{\tau_i}|a_i|$.
We first show that $W(a)$ bounds the length of the projection of $a$
onto $L$. Since $P$ is an orthogonal projector, it satisfies
$P^\top=P$ and $P^2=P$. Consequently,
\[
\norm{Pe_i}_2^2
=e_i^\top P^\top Pe_i
=e_i^\top Pe_i
=P_{ii}
=\tau_i.
\]
Writing $a=\sum_i a_i e_i$ and applying the triangle inequality gives
\[
\norm{Pa}_2
=\left\lVert\sum_i a_i Pe_i\right\rVert_2
\le\sum_i |a_i|\norm{Pe_i}_2
=\sum_i\sqrt{\tau_i}|a_i|
=W(a).
\]
Thus the weighted absolute residual controls the part of the
residual that lies in the prediction direction space.

By definition, our mechanism minimizes $W(y-z)$ over feasible
predictions $z\in K$. Since the least-squares prediction $z^\star$
is also feasible, optimality gives
$W(\widehat r)\le W(r^\star)$.
Moreover, the leverage scores sum to the dimension of $L$:
$\sum_i\tau_i=\operatorname{tr}(P)=k$.
Cauchy--Schwarz therefore implies
\[
W(r^\star)
=\sum_i\sqrt{\tau_i}|r_i^\star|
\le
\left(\sum_i\tau_i\right)^{1/2}
\left(\sum_i(r_i^\star)^2\right)^{1/2}
=\sqrt{k}\,\norm{r^\star}_2.
\]
Combining these inequalities yields $\norm{P\widehat r}_2 \le W(\widehat r) \le W(r^\star) \le \sqrt{k}\,\norm{r^\star}_2$. We conclude
\begin{equation*}
\norm{\widehat r}_2^2 =\norm{P\widehat r}_2^2 + \norm{(I-P)\widehat r}_2^2 \le k\norm{r^\star}_2^2+\norm{r^\star}_2^2 =(k+1)\norm{r^\star}_2^2 =(k+1)\OPT_\Theta(y). \qedhere
\end{equation*}

\end{proof}

The lower bound in \Cref{sec:lower} shows that the unrestricted factor $d+1$ is optimal even among universally strategyproof randomized mechanisms. 

\section{Strategyproofness in expectation}
\label{sec:spie}
We now introduce randomness and let each agent minimize expected $p$-th power error, $\E|Z_i-y_i|^p$, for some $p>0$. The social objective remains squared error. Throughout this section, the hypothesis class is \emph{complete}: $\Theta=\R^d$, so every vector in $L=\col(X)$ is an admissible prediction. We first show a sharp bound for the perhaps most natural choice, $p=1$, then demonstrate that randomization cannot help for $p=2$, when the individual loss and the social loss coincide, and finally show that every other positive exponent permits a strict improvement over $d+1$. Together, these results give an exact characterization of when randomization helps.

\subsection{Absolute individual loss: a sharp upper bound at \texorpdfstring{$p=1$}{p=1}}
\label{sec:randomized}

A natural randomized approach is a form of \emph{random dictatorship}: randomly select a group of agents and interpolate their reported data. Full interpolation of $d$ independent observations protects the sampled agents but leaves no freedom to fit the remaining data. Interpolating fewer reports preserves useful freedom, but requires extra care to ensure truthfulness.

\paragraph{Mechanism: Partial interpolation.} 
Let $L=\col(X)$ and $H=X(X^\top X)^{-1}X^\top$, the orthogonal projector onto $L$. The ordinary least-squares (OLS) predictions are $Hy$, and the residual is $r^\star=(I-H)y$. Thus $r^\star\perp L$ and $\OPT(y)=\norm{r^\star}_2^2$. For $B\subseteq[n]$, $H_B$ denotes its principal submatrix, $H_{:,B}$ its columns, and $y_B$ the restricted vector.

For $b\in\{0,\ldots,d\}$, draw a set $B\subseteq[n]$ of size $b$ with the probability below and return
\begin{equation}
\Prb(B)=\frac{\det(H_B)}{\binom db},
\qquad
Z^B(y)=\argmin_{z\in L:\ z_B=y_B}\norm{z-y}_2^2.
\label{eq:partial-interpolation}
\end{equation}
The probabilities depend only on public features. Geometrically, the distribution favors observations spanning different directions in the prediction space. We may randomize $b$ independently of the reports. When $b=0$, the mechanism is the (non-strategyproof) OLS rule; when $b=d$, it is the $(d+1)$-efficient volume sampling mechanism of \citet{derezinski2017}.

We give the probabilities explicitly for ease of interpretation, but sampling does not require enumerating subsets. An iterative sampling scheme, adapted from the projection determinantal point process (DPP) sampler presented by \citet{kulesza2012}, takes $O(nd^2)$ arithmetic operations; see \Cref{app:iterative-sampling}. 

\paragraph{Efficiency and strategyproofness.} We first note that \eqref{eq:partial-interpolation} indeed defines a probability distribution. Since $H$ is a rank-$d$ orthogonal projector, its principal-submatrix determinants satisfy
\begin{equation}
    \sum_{|B|=b}\det(H_B) =\binom db,
    \label{eq:minor-normalization}
\end{equation}
as shown by the determinant expansion in \Cref{app:partial-spie}. These determinants are nonnegative, so dividing by $\binom db$ gives probabilities that sum to one.

The following two lemmas establish the approximation guarantee and truthfulness. 

\begin{lemma}[Partial-interpolation efficiency]
\label[lemma]{lem:partial-efficiency}
For the complete hypothesis class $\Theta=\R^d$, the fixed-size rule \eqref{eq:partial-interpolation} with parameter $b$ has squared-error approximation ratio at most $(d+1)/(d-b+1)$.
\end{lemma}

\begin{proof}
The case $b=0$ returns OLS, so suppose $b\ge1$. Write $r^\star=(I-H)y$. We first show that least squares subject to interpolating the reports on $B$ produces the prediction $Hy+q^B$, where $q^B:=H_{:,B}H_B^{-1}r_B^\star$.

To interpolate the sampled reports, we add
$q^B$ to the OLS prediction $Hy$. This correction belongs to $L$ and satisfies $(q^B)_B=r_B^\star$. It is also orthogonal to every $v\in L$ with $v_B=0$, since
$\langle q^B,v\rangle=(r_B^\star)^\top H_B^{-1}v_B=0$.
Every other feasible correction has the form $q^B+v$, where $v\in L$ and $v_B=0$. Orthogonality gives $\norm{q^B+v}_2^2=\norm{q^B}_2^2+\norm v_2^2$, so $q^B$ has minimum norm. Because $r^\star\perp L$ and $H^2=H$, we obtain
\begin{equation}
Z^B(y)=Hy+q^B,\qquad
\norm{Z^B(y)-y}_2^2
=\norm{r^\star}_2^2+(r_B^\star)^\top H_B^{-1}r_B^\star.
\label{eq:constrained-loss-main}
\end{equation}

We now bound the average of $(r_B^\star)^\top H_B^{-1}r_B^\star$, weighted by $\det(H_B)$. A rank-one perturbation captures this product: for each sampled $B$ and $t>0$, the determinant lemma (e.g. \citet[p.~26, Eq.~(0.8.5.11)]{horn2013}) gives
\begin{equation}
\det(H_B+t r_B^\star(r_B^\star)^\top)
=\det(H_B)+t\det(H_B)(r_B^\star)^\top H_B^{-1}r_B^\star.
\label{eq:determinant-lemma}
\end{equation}
The coefficient of $t$ is the extra error. To sum these determinants, recall that $r^\star\perp L$: adding $t r^\star(r^\star)^\top$ leaves the $d$ unit eigenvalues of $H$ unchanged and adds the eigenvalue $t\norm{r^\star}_2^2$ in the residual direction when $r^\star\ne0$. Each nonzero product of $b$ eigenvalues either uses only unit eigenvalues or uses this additional eigenvalue and $b-1$ unit eigenvalues. Thus,
\begin{equation}
\sum_{|B|=b}\det(H_B+t r_B^\star(r_B^\star)^\top)
=\binom db+t\binom d{b-1}\norm{r^\star}_2^2.
\label{eq:rank-one-minors-main}
\end{equation}
The formula also holds when $r^\star=0$. Summing \eqref{eq:determinant-lemma} over sampled sets gives $\binom db+t\binom db\E_B[(r_B^\star)^\top H_B^{-1}r_B^\star]$. This is at most the full sum in \eqref{eq:rank-one-minors-main}: the remaining minors are nonnegative because $H+t r^\star(r^\star)^\top$ is positive semidefinite. Subtracting $\binom db$ and dividing by $t\binom db$ therefore gives
\[
\E_B[(r_B^\star)^\top H_B^{-1}r_B^\star]
\le\frac{\binom d{b-1}}{\binom db}\norm{r^\star}_2^2
=\frac{b}{d-b+1}\norm{r^\star}_2^2.
\]
Combining this with \eqref{eq:constrained-loss-main} and
$\norm{r^\star}_2^2=\OPT(y)$ proves the claim.
\end{proof}

\begin{lemma}[Partial-interpolation strategyproofness]
\label[lemma]{lem:partial-spie}
A report-independent distribution of $b$ with $\E b\ge d/2$ makes the mechanism SPIE for absolute individual loss.
\end{lemma}

\begin{proof}[Proof Sketch]
Fix an agent $i$, the other reports, and a deviation from their true label $y_i$ to $y_i+\delta$. When $i\in B$, the mechanism interpolates their report, so the deviation increases their loss from zero to $|\delta|$. When $i\notin B$, the agent can influence only the part of the fit that remains free after interpolation.

To quantify this remaining influence, let $H^B$ be the orthogonal projector onto $\{z\in L:z_B=0\}$, and set $h_i^B=(H^B)_{ii}\ge0$. The constrained least-squares formula in \eqref{eq:constrained-loss-main} shows that an unselected agent's prediction changes by $h_i^B\delta$. Their absolute loss can therefore decrease by at most $h_i^B|\delta|$.

For each fixed sample size $b$, we prove in \Cref{app:partial-spie} that the sampling rule gives
\begin{equation}
\Prb(i\in B)=\frac bd H_{ii},
\qquad
\E_B[\mathbf{1}_{\{i\notin B\}}h_i^B]=\frac{d-b}{d}H_{ii}.
\label{eq:balance}
\end{equation}
Intuitively, the first formula allocates the $b$ interpolation slots in proportion to each agent's leverage $H_{ii}$, whose total is $d$. For the second formula, the $b$ interpolation constraints fix $b$ of the $d$ prediction directions. The determinant weights leave, on average, a fraction $(d-b)/d$ of each agent's original influence in the remaining fit. Thus larger samples make an agent more likely to be selected while reducing the expected influence remaining outside the sampled set. 

Let $\Delta_i^B$ be the agent's loss after the deviation minus their truthful loss. Since the sampling distribution is independent of reports, the two cases above give
\[
\E_B\Delta_i^B
\ge |\delta|\left(\Prb(i\in B)-\E_B[\mathbf{1}_{\{i\notin B\}}h_i^B]\right)
=\frac{2b-d}{d}H_{ii}|\delta|.
\]
Averaging over the report-independent choice of $b$ replaces $b$ by $\E b$. The expected loss cannot decrease when $\E b\ge d/2$, which proves strategyproofness in expectation.
\end{proof}

To obtain the final bound, take $b=d/2$ when $d$ is even, and mix $b=(d-1)/2$ and $b=(d+1)/2$ equally when $d$ is odd. In both cases $\E b=d/2$, so \Cref{lem:partial-spie} gives SPIE. Averaging the corresponding bounds in \Cref{lem:partial-efficiency} gives the ratio below.

\begin{theorem}[A uniform sub-$2$ upper bound]
\label{thm:absolute}
For absolute individual loss and the complete hypothesis class $\Theta=\R^d$, an SPIE mechanism achieves squared-error approximation ratio
\begin{equation*}
\rho_d=2-\frac{1}{\lceil d/2\rceil+1}.
\label{eq:rho}
\end{equation*}

\end{theorem}
As an example, in constant fitting ($d=1$), the rule returns the mean with probability $1/2$ and a uniformly chosen report otherwise, recovering the optimal $3/2$ mechanism of \citet{feldman2013}. 


\paragraph{The absolute-loss bound is tight.} Finally, we prove tightness of the above result. 

\label{sec:absolute-lower}
\begin{theorem}[Optimality under absolute individual loss]
\label{thm:absolute-lower}
For every $d\ge1$, there is a full-column-rank design with $\Theta=\R^d$ on which every absolute-loss SPIE mechanism has squared-error approximation ratio at least $\rho_d$. 
\end{theorem}

\begin{proof}[Proof Sketch]
Set $n=d+1$ and take $X$ with first $d$ rows $I_d$ and last row $-\ones^\top$, so predictions sum to zero. For labels summing to $t>0$, residuals therefore sum to $t$. Least squares spreads this residual equally, giving $\OPT(y)=t^2/n$. For a $c$-approximation, the normalized residual $V=(y-Z(y))/t$ satisfies $\sum_i V_i=1$ and $\E\norm V_2^2\le c/n$.

Accuracy favors spreading this unit residual over many agents. Let $P$ count the positive coordinates of $V$, and let $A$ be its total negative mass. The positive mass is $1+A$, so Cauchy--Schwarz gives $\norm V_2^2\ge(1+A)^2/P$. Thus few positive coordinates are costly, and negative residuals require still more positive mass. Accounting for the fact that $P$ is an integer, \Cref{app:absolute-lower} shows that beating $\rho_d$ would require
\[
\E[P-A]\ge\frac{n+1}{2}+\Delta
\]
for some constant $\Delta>0$, uniformly over all profiles with $t>0$.

Incentives constrain this spreading: truthfulness gives a lower bound $2\Prb(V_i>0)-1$ on the derivative of agent $i$'s expected absolute loss with respect to their own label, wherever it exists. If we could sum these rates along a common path of profiles with total $t$, the total expected absolute loss $F(t)$ would satisfy
\[
F'(t)\ge 2\E P-n\ge\frac{F(t)}t+2\Delta,
\]
since $F(t)/t=\E\norm V_1=1+2\E A$. This gives $(F(t)/t)'\ge2\Delta/t$, forcing normalized loss to grow without bound, whereas accuracy requires $F(t)/t\le\sqrt c$.

The agents' report paths do differ, but averaging over a growing ball in the zero-sum prediction space makes the common-path calculation accurate up to a vanishing error on each fixed $t$-interval. \Cref{app:absolute-lower} justifies this step, completing the contradiction and proving $c\ge\rho_d$.
\end{proof}

\subsection{Squared individual loss: a lower bound at \texorpdfstring{$p=2$}{p=2}}
\label{sec:lower}
Next, we consider what happens when the individual loss has the same form as the social loss and each agent aims to minimize their expected squared loss. Perhaps surprisingly, unlike with absolute individual loss, randomization provides no improvement over deterministic mechanisms in this case.

\begin{theorem}[Optimality under squared individual loss]
\label{thm:lower}
For each $d\ge1$, there is a full-column-rank design with $\Theta=\R^d$ such that every SPIE mechanism for squared individual loss has approximation ratio at least $d+1$. 
\end{theorem}

\begin{proof}[Proof Sketch]
Use the same design as in \Cref{thm:absolute-lower}, with $n=d+1$ and $L=\ones^\perp$. For a profile $y$ with $\sum_i y_i=\tau>0$, the optimum is $\tau^2/n$.

Suppose the mechanism has finite approximation ratio $c$. Reporting $y_i-\tau$ makes the profile $y-\tau e_i$ realizable, so it must be fitted exactly. Starting from this report and increasing it to $y_i$, the squared-loss incentive inequalities give
\begin{equation*}
C_i(y)=2\int_0^\tau r_i\bigl(y-(\tau-s)e_i\bigr)\,ds,
\label{eq:envelope-main}
\end{equation*}
where $C_i(v)=\E[(Z_i(v)-v_i)^2]$ and $r_i(v)=v_i-\E Z_i(v)$. At parameter $s$, agent $i$ reports $y_i-\tau+s$, so the sum of all reports is $s$.

If the integrands were evaluated at a common profile $v(s)$ with $\sum_jv_j(s)=s$, feasibility would give $\sum_i r_i(v(s))=s$. Their sum would therefore be
\[
2\int_0^\tau\sum_i r_i(v(s))\,ds
=2\int_0^\tau s\,ds=\tau^2,
\]
which is $n$ times the optimum. 

Again, different agents follow different paths, but the same translation averaging from \Cref{thm:absolute-lower} justifies summing: the paths differ by bounded shifts within $L$, whose contribution vanishes when averaged over larger balls. Thus the averaged social cost tends to $\tau^2$, while remaining at most $c\tau^2/n$, giving $c\ge n=d+1$. Unlike absolute loss, squared loss determines this sum at each fixed $\tau$, and no growth argument over $\tau$ is needed. \Cref{app:lower} supplies the details.
\end{proof}

The same worst-case lower bound holds for affine predictors on $\R^{d-1}$ with a free intercept as well, as detailed in \Cref{app:affine-lower}. Thus, the optimal worst-case ratio is $d+1$ for squared-loss SPIE. Since all positive powers induce the same deterministic preferences, every universally strategyproof mechanism is also squared-loss SPIE, demonstrating the optimality of \Cref{thm:deterministic}.

\subsection{Other powers: an upper bound for every \texorpdfstring{$p\ne2$}{p not equal to 2}}
\label{sec:powers}

In general, do SPIE mechanisms improve on deterministic mechanisms? In other words, is absolute loss special, or is squared loss the exception? The next construction shows that it is the latter. It combines two independent interpolating fits through a small lottery. Averaging the fits improves accuracy; additional outcomes make truthful reporting optimal. The hypothesis class is again complete, $\Theta=\R^d$.

\begin{theorem}[Randomization helps at every positive power except two]
\label{thm:powers}
For every $p>0$ and every full-column-rank public design with $d$ free coefficients, the complete hypothesis class $\Theta=\R^d$ admits an SPIE mechanism with squared-error approximation ratio at most $1+dq_p$, where, for $\kappa=1+\sqrt2$,
\begin{equation}
q_p=
\begin{cases}
\dfrac{2+2^{p-1}}{2+2^p},&0<p<1,\\[5pt]
\dfrac12+\lambda_p,\quad
\lambda_p=\dfrac{2}{2+\kappa^{p-2}+\kappa^{2-p}},&p\ge1.
\end{cases}
\label{eq:power-factor}
\end{equation}
Here $q_2=1$ and $q_p<1$ for every $p\ne2$.
\end{theorem}

\begin{proof}[Proof Sketch] 
First, draw $B$ and $C$ independently using \eqref{eq:partial-interpolation} with $b=d$. The interpolating fits are $U=XX_B^{-1}y_B$ and $V=XX_C^{-1}y_C$. Singular sets have probability zero. These fits satisfy
\begin{equation}
\E U=Hy,\qquad
\sigma^2:=\E\norm{U-Hy}_2^2\le d\,\OPT(y).
\label{eq:two-fit-moments}
\end{equation}
The same statements hold for $V$. The second-moment bound follows from \Cref{lem:partial-efficiency} at $b=d$ and the orthogonality of $U-Hy$ and $y-Hy$. The unbiasedness formula $\E U=Hy$ is proved in \Cref{powers:volume-moments}. For $b=d$, these properties were also proved by \citet{derezinski2017} in the context of learning from small samples.

We consider $p\ge1$, where the loss is convex, and $0<p<1$, where subadditivity replaces convexity.

First suppose $p\ge1$. Put $M=(U+V)/2$ and return
\begin{equation}
Z=\begin{cases}
M,&\text{with probability }1-\lambda_p,\\
M+(U-V)/\sqrt2,&\text{with probability }\lambda_p/2,\\
M-(U-V)/\sqrt2,&\text{with probability }\lambda_p/2.
\end{cases}
\label{eq:power-lottery}
\end{equation}
The midpoint reduces variance, while the two extrapolated outcomes enforce incentives. To see the incentive balance, write $Z=(1-T)U+TV$, with $T\in\{1/2,(1-\sqrt2)/2,(1+\sqrt2)/2\}$. Conditional on $B,C$, an agent in neither sample has no influence, while an agent in both is fitted exactly. An agent in exactly one sample trades off the fit that interpolates their report against a fit independent of it. Their expected $p$-th power loss is convex in their report. The chosen $\lambda_p$ makes its derivative at truth zero when the two truthful predictions at that agent differ; if they coincide, truth has zero loss. Independence and unbiasedness give squared social cost $\OPT(y)+(\tfrac12+\lambda_p)\sigma^2$, proving the convex-range bound.

For $0<p<1$, convexity no longer holds. Instead return $U$ and $V$ each with probability $w_p=1/(2+2^p)$, and their midpoint with probability $2^pw_p$. If only one fit interpolates the agent, writing $s$ for the other fit's error and $\delta$ for their deviation gives expected loss $w_p(|\delta|^p+|s|^p+|s+\delta|^p)$. The inequality $|s|^p\le|\delta|^p+|s+\delta|^p$ makes truth optimal. The same variance calculation gives $q_p$ in \eqref{eq:power-factor}.
\end{proof}

\paragraph{Why is $p=2$ so different?}
Squared loss is special because the moment that balances incentives is exactly the moment that controls the benefit of averaging. To see this in the two-fit construction, condition on $B,C$ and consider an agent in $B\setminus C$; the other case is symmetric. If the other fit's prediction differs from their true value by $s\ne0$, their expected squared loss after a deviation $\delta$ is
$F(\delta)=\E[(T s+(1-T)\delta)^2]$.
For truthfulness conditional on the samples, we need
$F'(0)=2s\,\E[T(1-T)]=0$. But the variance multiplier of the combined fit is
\[
\E[(1-T)^2+T^2]
=1-2\E[T(1-T)]
=1.
\]
Thus averaging cannot reduce variance under this conditional incentive balance: the extrapolation needed for truthfulness exactly cancels the gain from the midpoint. In our lottery, this balance gives $\lambda_2=1/2$ and $q_2=1$.

\section{Conclusion and Discussion}
\label{sec:discussion}

The optimal worst-case deterministic and universally strategyproof price is $d+1$ for all single-peaked preferences, where $d$ counts free coefficients rather than observations. For SPIE, squared individual loss is the unique power that retains this price. Absolute loss has the exact optimum $\rho_d=2-1/(\lceil d/2\rceil+1)<2$, while each $p>0$ with $p\notin\{1,2\}$ also admits a strict improvement over $d+1$. 

Our SPIE mechanisms assume the complete hypothesis set. An immediate open problem is to obtain similar guarantees for arbitrary convex and compact hypothesis sets. A more ambitious question is to give a complete characterization of strategyproof mechanisms. Other interesting directions include going beyond linear regression, such as to non-linear models or deep learning, or incorporating external predictions and advice.

\label{main-text-end}

\clearpage
\subsection*{AI use statement}
ChatGPT 5.6-Sol and ChatGPT 6-Astra were used to assist with exploring proof sketches, identifying and checking references, and drafting the manuscript. For \cref{sec:deterministic}, the authors first considered the complete hypothesis set and proposed weighted $L_1$-ERM, with weights reflecting how strongly changes in the model can affect each agent's prediction. GPT derived the $(d+1)$ approximation guarantee and subsequently generalized the argument to arbitrary convex sets. For \cref{sec:spie}, the authors proposed volume sampling for the special case $b=d$, and GPT subsequently suggested selecting smaller values of $b$. For the lower bound, the authors first constructed a $(d+1)$ lower bound for all deterministic strategyproof mechanisms, and GPT extended it to SPIE mechanisms with squared loss at the authors' request. The authors are responsible for the correctness, originality, and accuracy of the final submission, including all AI-assisted text and mathematical arguments.

\bibliography{refs.bib}
\bibliographystyle{iclr2027_conference}

\clearpage
\appendix
\crefalias{section}{appendix}
\crefalias{subsection}{appendix}
\section{Implementation and analysis of partial interpolation}
\label{app:partial-spie}

We first give an iterative implementation of the sampling rule, then explain how an unselected report affects the fit and prove the two formulas in \eqref{eq:balance}. Together, these calculations complete the proof of \Cref{lem:partial-spie}. Throughout, $H$ is the rank-$d$ orthogonal projector onto $L$, and a size-$b$ set $B$ has probability $\det(H_B)/\binom db$.

\subsection{Iterative sampling and running time}
\label{app:iterative-sampling}

We use the Gram--Schmidt projection-DPP sampler presented by \citet[Section~2.4.4, Algorithm~1]{kulesza2012} and stop after $b$ selections. The calculation below shows that this yields the distribution in \eqref{eq:partial-interpolation}.

Compute a thin QR factorization $X=UR$, where $U\in\R^{n\times d}$ has orthonormal columns, so $H=UU^\top$. Write $u_i^\top$ for row $i$ of $U$. Initialize $v_i=u_i$ for every $i$ and $B=\varnothing$. For $t=0,\ldots,b-1$, perform the following steps:
\begin{enumerate}
\item Select an index $i_t$ with probability $\norm{v_{i_t}}_2^2/(d-t)$ and add it to $B$.
\item Set $q_t=v_{i_t}/\norm{v_{i_t}}_2$, and update every vector by $v_j\leftarrow v_j-\langle v_j,q_t\rangle q_t$.
\end{enumerate}
Return the unordered set $B$. For $b=0$, return the empty set. Each update removes the direction just selected, so subsequent selections favor observations that contribute a new direction. Previously selected rows have zero residual and cannot be selected again.

\paragraph{Correctness.}
The vectors $q_0,\ldots,q_{t-1}$ are orthonormal. At the start of step $t$, let $P_t=I_d-\sum_{s<t}q_sq_s^\top$, so $v_i=P_tu_i$. Since $P_t$ is a rank-$(d-t)$ projector and $U^\top U=I_d$, the probabilities at that step sum to one:
\[
\sum_{i=1}^n\norm{v_i}_2^2
=\operatorname{tr}(P_tU^\top UP_t)
=\operatorname{tr}(P_t)
=d-t.
\]
Let $U_B$ contain the selected rows in their selection order. Their coordinates in the orthonormal directions $q_0,\ldots,q_{b-1}$ form a lower-triangular matrix, with the residual lengths at selection on its diagonal. Thus $\det(U_BU_B^\top)=\det(H_B)$ is the product of the squared residual lengths. Multiplying the stepwise selection probabilities therefore gives, for $b\ge1$,
\[
\Prb(i_0,\ldots,i_{b-1})
=\frac{\det(H_B)}{d(d-1)\cdots(d-b+1)}.
\]
There are $b!$ orderings, so their probabilities sum to $\det(H_B)/\binom db$, as required.

\paragraph{Computational cost.}
The QR factorization costs $O(nd^2)$ arithmetic operations and can be reused because $X$ is public and fixed. Each sampling step updates $n$ vectors of length $d$, so drawing $B$ costs $O(ndb)$ after preprocessing. The fit can also be computed without forming the $n\times n$ matrix $H$: set $a=U^\top y$, solve $(U_BU_B^\top)c=y_B-U_Ba$, and return $U(a+U_B^\top c)$. This is the prediction in \eqref{eq:constrained-loss-main} and costs $O(nd+db^2+b^3)$ operations; when $b=0$, simply return $Ua$. Since $b\le d\le n$, the total cost is $O(nd^2)$. The two independent full-size draws used in \Cref{sec:powers} have the same asymptotic running time.

\subsection{How an unselected report affects the fit}
For a sampled nonempty set $B$, the projector onto the prediction directions that vanish on $B$ is
\begin{equation*}
H^B=H-H_{:,B}H_B^{-1}H_{B,:},
\qquad h_i^B=(H^B)_{ii}.
\label{eq:remaining-projector}
\end{equation*}
To see this, $H^2=H$ gives $H_{B,:}H_{:,B}=H_B$, so the subtracted matrix is the orthogonal projector onto the span of the columns $He_j$, $j\in B$. Within $L$, being orthogonal to these columns is equivalent to having $z_j=0$ for all $j\in B$, since $\langle He_j,z\rangle=z_j$. Thus $H^B$ retains exactly the allowed directions of change after interpolation, and $h_i^B\ge0$.

For $i\notin B$, changing $y_i$ by $\delta$ leaves $y_B$ unchanged. Substituting into the prediction formula in \eqref{eq:constrained-loss-main} gives $Z^B(y+\delta e_i)-Z^B(y)=\delta H^Be_i$. In particular, agent $i$'s prediction changes by $h_i^B\delta$. For $b=0$, take $H^\varnothing=H$: the agent is never selected and retains influence $H_{ii}$, proving both formulas in \eqref{eq:balance} in this case. Below, assume $1\le b\le d$.

\subsection{Selection probabilities and remaining influence}

\paragraph{A determinant expansion.}
For a symmetric $m\times m$ matrix $A$, let $A_C$ retain the rows and columns indexed by the same set $C$; its determinant is called a \emph{principal minor}, with $\det(A_\varnothing)=1$. If $\lambda_1,\ldots,\lambda_m$ are the eigenvalues of $A$, then
\begin{equation}
\sum_{\substack{C\subseteq[m]\\|C|=s}}\det(A_C)
=\sum_{1\le j_1<\cdots<j_s\le m}\lambda_{j_1}\cdots\lambda_{j_s},
\qquad s\ge1.
\label{eq:principal-minor-sum}
\end{equation}
For $1\le s\le m$, this is \citet[Theorem~1.2.16, p.~54]{horn2013}; for $s>m$, both sums are empty and equal zero. To prove it, expand $\det(I_m+tA)$ by linearity in each column. Choosing columns from $tA$ exactly at the indices in $C$ contributes $t^{|C|}\det(A_C)$, since the remaining columns come from $I_m$. On the other hand, the eigenvalues give $\det(I_m+tA)=\prod_{j=1}^m(1+t\lambda_j)$. Comparing coefficients of $t^s$ proves \eqref{eq:principal-minor-sum}. For $A=H$, only products of $s$ of its $d$ unit eigenvalues are nonzero, giving the normalization in \eqref{eq:minor-normalization}; the case $b=0$ has the single empty set of weight one.

\paragraph{Selection probability.}
Using the matrix $U$ from \Cref{app:iterative-sampling}, let $U_{-i}$ denote $U$ with row $i$ removed. The matrix obtained by deleting row and column $i$ from $H$ is $U_{-i}U_{-i}^\top$. Its nonzero eigenvalues coincide with those of
\[
U_{-i}^\top U_{-i}=I_d-u_i u_i^\top.
\]
The latter has $d-1$ eigenvalues equal to $1$ and one equal to $1-\norm{u_i}_2^2=1-H_{ii}$. Zero eigenvalues do not affect the sums in \eqref{eq:principal-minor-sum}. A product of $b$ eigenvalues therefore either uses only unit eigenvalues or uses $1-H_{ii}$ together with $b-1$ unit eigenvalues. Consequently,
\begin{align*}
\sum_{\substack{|B|=b\\i\notin B}}\det(H_B)
&=\binom{d-1}{b}+(1-H_{ii})\binom{d-1}{b-1},\\
\sum_{\substack{|B|=b\\i\in B}}\det(H_B)
&=\binom db-\sum_{\substack{|B|=b\\i\notin B}}\det(H_B)
=\binom{d-1}{b-1}H_{ii}.
\label{eq:minor-containing-i}
\end{align*}
Here a binomial coefficient is zero when its lower index exceeds its upper index. Dividing the second equality by $\binom db$ proves $\Prb(i\in B)=(b/d)H_{ii}$.

\paragraph{Expected remaining influence.}
For a sampled set $B$ with $i\notin B$, order the rows and columns of $H_{B\cup\{i\}}$ with $i$ last. Subtracting $H_{i,B}H_B^{-1}$ times its first $b$ rows from its last row gives
\begin{equation*}
\det(H_{B\cup\{i\}})
=\det(H_B)\bigl(H_{ii}-H_{i,B}H_B^{-1}H_{B,i}\bigr)
=\det(H_B)h_i^B.
\label{eq:schur-leverage}
\end{equation*}
Thus the sampling weight times the remaining influence is the determinant for the enlarged set $C=B\cup\{i\}$. If $\det(H_C)>0$, then $H_C$ and its principal submatrix $H_{C\setminus\{i\}}$ are positive definite. Hence enlarged sets not arising from a positive-weight $B$ contribute zero, and we may sum over all $C$. For $b<d$, the selection calculation for sets of size $b+1$ gives
\begin{align*}
\E_B[\mathbf{1}_{\{i\notin B\}}h_i^B]
&=\frac{1}{\binom db}\sum_{\substack{|C|=b+1\\i\in C}}\det(H_C)\\
&=\frac{\binom{d-1}{b}}{\binom db}H_{ii}
=\frac{d-b}{d}H_{ii}.
\end{align*}
For $b=d$, every enlarged determinant is zero because $H$ has rank $d$, giving the same formula.

Finally, a deviation by $\delta$ costs $|\delta|$ when the agent is selected and can reduce their loss by at most $h_i^B|\delta|$ otherwise. The two formulas in \eqref{eq:balance} therefore give an expected loss increase of at least $(2\E b-d)H_{ii}|\delta|/d$. This is nonnegative under the condition of \Cref{lem:partial-spie}, completing the proof.

\section{The lower bound for absolute individual loss}
\label{app:absolute-lower}

We prove \Cref{thm:absolute-lower} by summing the agents' incentive constraints. Set $n=d+1$ and take the design with first $d$ rows $I_d$ and last row $-\ones^\top$, whose prediction space is $L=\ones^\perp$. Write $y(w,t)=w+(t/n)\ones$ for $w\in L$ and $t>0$. Then $\OPT(y(w,t))=t^2/n$.

The argument has two parts. A geometric inequality shows how many positive residuals a sufficiently accurate mechanism must produce. Truthfulness then forces its absolute loss to grow too quickly. Averaging over translations $w$ lets us compare the agents' report paths at a common $t$.

\subsection{Accuracy requires many positive residuals}

For a vector $v$ with $\sum_i v_i=1$, define $P(v)=|\{i:v_i>0\}|$ and $A(v)=\sum_{i:v_i<0}(-v_i)$. Thus $1\le P(v)\le n$, its positive coordinates sum to $1+A(v)$, and $\norm v_1=1+2A(v)$.

\begin{lemma}[A geometric residual bound]
\label[lemma]{lem:absolute-moment}
Let $k=\lfloor(n+1)/2\rfloor$. Every $v\in\R^n$ with $\sum_i v_i=1$ satisfies
\begin{equation*}
P(v)-A(v)\ge 2k+1-k(k+1)\norm v_2^2.
\label{eq:absolute-pointwise-moment}
\end{equation*}
\end{lemma}

\begin{proof}
Cauchy--Schwarz applied to the positive coordinates gives $\norm v_2^2\ge(1+A(v))^2/P(v)$. To sharpen this bound, note that for every positive integer $p$,
\[
\frac1p-\frac{2k+1-p}{k(k+1)}
=\frac{(p-k)(p-k-1)}{pk(k+1)}\ge0.
\]
The last inequality holds because no integer lies strictly between $k$ and $k+1$. Also, $p\le n\le2k(k+1)$ implies $2/p\ge1/[k(k+1)]$. Consequently,
\[
\norm v_2^2
\ge\frac{1+2A(v)}{P(v)}
\ge\frac{2k+1-P(v)+A(v)}{k(k+1)}.
\]
Rearranging proves the claim.
\end{proof}

\subsection{Aligning the agents' report paths}
\label{app:absolute-paths}

Let $Z$ be a measurable absolute-loss SPIE mechanism with approximation ratio $c<\infty$. Write $r(y)=y-Z(y)$ for its random residual and $g_i(y)=\E|r_i(y)|$ for agent $i$'s truthful expected loss. For $R>0$, let $B_R=\{w\in L:\norm w_2\le R\}$ and let $W_R$ be uniform on $B_R$. For $t>0$, define
\[
Y_R(t)=W_R+\frac tn\ones,\qquad
F_R(t)=\E\norm{r(Y_R(t))}_1,\qquad
N_R(t)=\sum_{i=1}^n\Prb(r_i(Y_R(t))>0).
\]
These expectations and probabilities include both $W_R$ and the mechanism's randomness. Since every residual sums to $t$, the triangle inequality and the approximation guarantee give
\begin{equation}
t\le F_R(t)\le\sqrt{n\,\E\norm{r(Y_R(t))}_2^2}\le\sqrt c\,t.
\label{eq:absolute-l1-bounds}
\end{equation}

For a single agent, truthfulness controls how $g_i$ changes with their own report. We would like to add these changes at the same profile. The next lemma makes this step precise: averaging over $B_R$ aligns the paths up to an error that vanishes as $R$ grows.

\begin{lemma}[Averaged incentive inequality]
\label[lemma]{lem:absolute-averaging}
The function $F_R$ is locally Lipschitz on $(0,\infty)$, and there is a constant $C_{n,c}$ depending only on $n$ and $c$ such that
\begin{equation}
F_R'(t)\ge 2N_R(t)-n-\frac{C_{n,c}t}{R}
\quad\text{for almost every }t>0.
\label{eq:absolute-averaged-growth}
\end{equation}
\end{lemma}

\begin{proof}
\textbf{One agent's report.}
Fix the other reports and compare reports $a$ and $b$. Truthfulness and the triangle inequality give
\[
g_i(a,y_{-i})
\le\E|Z_i(b,y_{-i})-a|
\le g_i(b,y_{-i})+|a-b|.
\]
Exchanging $a$ and $b$ shows that the truthful expected loss is $1$-Lipschitz in the agent's own report.

We also need a one-sided bound. Apply truthfulness at the profile $y-he_i$, with agent $i$ deviating upward to $y_i$, where $h>0$. The deviating prediction is drawn from $Z(y)$, so
\begin{equation}
g_i(y)-g_i(y-he_i)
\ge\E\bigl[|r_i(y)|-|r_i(y)-h|\bigr].
\label{eq:absolute-finite-difference}
\end{equation}
Dividing by $h$ and letting $h\downarrow0$ would give the contribution $2\Prb(r_i(y)>0)-1$. The issue is that $y-he_i$ depends on $i$.

\textbf{Replacing shifted profiles by common profiles.}
Put $p_i=e_i-\ones/n\in L$. The two profiles in \eqref{eq:absolute-finite-difference} satisfy
\[
y(w,t)-he_i=y(w-hp_i,t-h).
\]
Thus all agents move between the same values $t$ and $t-h$, but their translations within $L$ differ. Averaging over a large ball makes this difference small.

To quantify it, write $\operatorname{vol}_L$ for volume in the $(n-1)$-dimensional space $L$. For $a\in L$ with $\norm a_2\le R$, the symmetric difference of $B_R$ and $B_R+a$ lies in the annulus with radii $R-\norm a_2$ and $R+\norm a_2$. Hence, with $\kappa_n=2^{n-1}(n-1)$,
\[
\frac{\operatorname{vol}_L(B_R\mathbin{\triangle}(B_R+a))}
{\operatorname{vol}_L(B_R)}
\le\frac{(R+\norm a_2)^{n-1}-(R-\norm a_2)^{n-1}}{R^{n-1}}
\le\kappa_n\frac{\norm a_2}{R}.
\]
Consequently, for any bounded measurable function $f:L\to\R$,
\begin{equation}
\left|\E f(W_R+a)-\E f(W_R)\right|
\le\kappa_n\sup_{w\in L}|f(w)|\,\frac{\norm a_2}{R}
\qquad(\norm a_2\le R).
\label{eq:absolute-ball-shift}
\end{equation}
The approximation guarantee supplies the required uniform bound:
\[
0\le g_i(y(w,s))
\le\sqrt{\E r_i(y(w,s))^2}
\le\sqrt{\frac cn}\,s
\qquad(w\in L,\ s>0).
\]

Let $G_{R,i}(t)=\E_{W_R}g_i(y(W_R,t))$, so $F_R=\sum_iG_{R,i}$. Since $\norm{p_i}_2\le1$, replacing $W_R-hp_i$ by $W_R$ at parameter $t-h$ costs at most $\kappa_n\sqrt{c/n}\,th/R$ when $0<h<\min\{t,R\}$. Combining this estimate with own-report Lipschitzness yields
\[
|G_{R,i}(t)-G_{R,i}(t-h)|
\le h+\kappa_n\sqrt{\frac cn}\,\frac{th}{R}.
\]
Thus each $G_{R,i}$, and therefore $F_R$, is locally Lipschitz and absolutely continuous on compact subintervals of $(0,\infty)$.

\textbf{Summing the incentive constraints.}
Average \eqref{eq:absolute-finite-difference} at $y=y(W_R,t)$ and use \eqref{eq:absolute-ball-shift} to obtain
\[
\frac{G_{R,i}(t)-G_{R,i}(t-h)}h
\ge
\E\!\left[\frac{|r_i(Y_R(t))|-|r_i(Y_R(t))-h|}{h}\right]
-\kappa_n\sqrt{\frac cn}\,\frac tR.
\]
For every real $x$, the quotient $(|x|-|x-h|)/h$ lies in $[-1,1]$ and converges to $2\mathbf1_{\{x>0\}}-1$ as $h\downarrow0$. Dominated convergence therefore gives, wherever $G_{R,i}'(t)$ exists,
\[
G_{R,i}'(t)
\ge2\Prb(r_i(Y_R(t))>0)-1
-\kappa_n\sqrt{\frac cn}\,\frac tR.
\]
Summing over agents proves \eqref{eq:absolute-averaged-growth} with $C_{n,c}=\kappa_n\sqrt{nc}$.
\end{proof}

\subsection{Completing the lower bound and transferring it to affine regression}
\label{app:absolute-affine}

\begin{proof}[Proof of \Cref{thm:absolute-lower}]
Set $k=\lfloor(n+1)/2\rfloor$. The target ratio can be written as
\begin{equation*}
\rho_{n-1}
=\frac{n(2k+1-(n+1)/2)}{k(k+1)}
=
\begin{cases}
2-\dfrac{2}{n+1},&n\text{ odd},\\[3pt]
2-\dfrac{2}{n+2},&n\text{ even}.
\end{cases}
\label{eq:absolute-critical-ratio}
\end{equation*}
Suppose for contradiction that the mechanism has ratio $c<\rho_{n-1}$, and define $\Delta=k(k+1)(\rho_{n-1}-c)/n>0$.

At each profile $y(w,t)$, apply \Cref{lem:absolute-moment} to the normalized residual $v=r(y(w,t))/t$. Since $\sum_i v_i=1$ and $\E\norm v_2^2\le c/n$, we get
\[
\E[P(v)-A(v)]
\ge2k+1-\frac{k(k+1)c}{n}
=\frac{n+1}{2}+\Delta.
\]
After averaging in $w$, the expected positive count is $N_R(t)$ and the expected normalized negative mass is $(F_R(t)/t-1)/2$. Therefore
\begin{equation}
2N_R(t)-n\ge\frac{F_R(t)}t+2\Delta.
\label{eq:absolute-count-growth}
\end{equation}
This is the conflict between accuracy and incentives: \eqref{eq:absolute-count-growth} requires many positive residuals, while each positive residual contributes to the loss growth in \eqref{eq:absolute-averaged-growth}.

Combining these two bounds, for almost every $t>0$,
\[
\left(\frac{F_R(t)}t\right)'
=\frac{F_R'(t)}t-\frac{F_R(t)}{t^2}
\ge\frac{2\Delta}{t}-\frac{C_{n,c}}R.
\]
Integrating from $1$ to any fixed $T>1$ gives
\[
\frac{F_R(T)}T-F_R(1)
\ge2\Delta\log T-\frac{C_{n,c}(T-1)}R.
\]
The left-hand side is at most $\sqrt c-1$ by \eqref{eq:absolute-l1-bounds}. First letting $R\to\infty$ with $T$ fixed, and then letting $T\to\infty$, gives a contradiction. Hence $c\ge\rho_{n-1}=\rho_d$, proving the claimed bound.

The same lower bound holds for ordinary affine regression. Use the construction in \Cref{app:affine-lower}: choose $d$ affinely independent feature points $x_1,\ldots,x_d\in\R^{d-1}$, place one agent at each, and place $d^2$ additional agents at their centroid $\bar x$. The coefficient space is the full $\R^d$, including the intercept. Map a label vector $y\in\R^{d+1}$ of the auxiliary design $\widetilde X$ to the vertex reports $y_1,\ldots,y_d$ and the common centroid report $u=-y_{d+1}/d$. As shown there, for every coefficient vector, the affine sum of squared residuals equals $\norm{\widetilde X\beta-y}_2^2$. The optima and approximation ratios therefore agree.

To transfer absolute-loss SPIE, deviations of the first $d$ auxiliary agents correspond to unilateral deviations of the vertex agents. For the last agent, fix true auxiliary label $t$ and alternative report $t'$, and put $u=-t/d$, $u'=-t'/d$. Change the $d^2$ centroid reports from $u$ to $u'$ one at a time. If $Q_j$ is the random prediction at the centroid after $j$ changes, individual truthfulness for the next agent, whose true value is $u$, gives
\begin{equation*}
\E|Q_{j-1}-u|\le\E|Q_j-u|\quad(1\le j\le d^2),
\qquad |-dQ_j-t|=d|Q_j-u|.
\label{eq:absolute-clone-transfer}
\end{equation*}
All centroid agents receive the same prediction. Hence these inequalities chain, and multiplication by $d$ gives the incentive inequality for the last auxiliary agent. Any affine mechanism with ratio below $\rho_d$ would therefore induce an absolute-loss SPIE mechanism with that ratio on $\widetilde X$, whose column space is $\ones^\perp$, contradicting the bound already proved.
\end{proof}

\paragraph{Where equality comes from.}
On the balanced design, $H=I-\ones\ones^\top/n$. Conditional on interpolating a set $B$ of size $b$, the least-squares residual has value zero on $B$ and $\tau/(n-b)$ elsewhere, where $\tau=\sum_i y_i$. For $\tau>0$, its normalized version therefore has $A=0$, $P=n-b$, and squared norm $1/(n-b)$. When $d$ is even, our mechanism uses $P=(n+1)/2$; when $d$ is odd, it mixes $P=n/2$ and $P=n/2+1$ equally. Thus the geometric bound in \Cref{lem:absolute-moment} is tight on this design. Moreover, $F_R(t)=t$ and $2N_R(t)-n=1$, so the incentive-growth calculation is also tight once the boundary error vanishes. This explains the exact parity-dependent factor.

\section{The lower bound for squared individual loss}
\label{app:lower}

We first prove the squared-loss lower bound for the balanced prediction space
and then transfer it to affine regression with \(d\) free coefficients,
including the intercept. Randomized mechanisms are assumed measurable.
The proof starts with an envelope formula for quadratic loss.

\begin{lemma}[Quadratic envelope formula]
\label[lemma]{lemma:envelope}
Fix the reports of all agents other than agent \(i\).  For a report
\(r\in\R\), let \(Z(r)\) be the random prediction received by agent \(i\), and
suppose that \(\E[Z(r)^2]<\infty\) for every \(r\in\R\). Define
\(\mu(r):=\E[Z(r)]\) and \(q(r):=\E[Z(r)^2]\).

If the mechanism is strategyproof in expectation for quadratic individual
loss, then \(\mu\) is nondecreasing.  Moreover, if \(Z(a)=a\) almost surely for
some \(a\in\R\), then, for every \(b\ge a\),

\begin{equation*}
\label{eq:lower-envelope}
    \E[(Z(b)-b)^2]
    =2\int_a^b \bigl(s-\mu(s)\bigr)\,ds.
\end{equation*}
\end{lemma}

\begin{proof}
For a true label \(t\) and an alternative report \(r\), strategyproofness in
expectation gives
\begin{equation}
\label{eq:lower-ic}
    q(t)-2t\mu(t)\le q(r)-2t\mu(r).
\end{equation}
Let
\(\Phi(t):=2t\mu(t)-q(t)\).  Equation~\eqref{eq:lower-ic}, together with the
fact that reporting \(r=t\) is feasible, shows that
\[
    \Phi(t)=\sup_{r\in\R}\{2t\mu(r)-q(r)\}.
\]
Thus \(\Phi\) is a finite convex function.  In addition, for all \(r,t\in\R\),
\begin{equation}
\label{eq:lower-subgradient}
    \Phi(t)\ge \Phi(r)+2(t-r)\mu(r),
\end{equation}

so \(2\mu(r)\) is a subgradient of \(\Phi\) at \(r\).  Monotonicity of
one-dimensional subgradients implies that \(\mu\) is nondecreasing.

A finite convex function on \(\R\) is absolutely continuous on every compact
interval and differentiable almost everywhere.  At every point \(s\) where
\(\Phi\) is differentiable,~\eqref{eq:lower-subgradient} implies
\(\Phi'(s)=2\mu(s)\). Integrating gives
\[
    \Phi(b)-\Phi(a)=2\int_a^b\mu(s)\,ds.
\]
Exactness at \(a\) gives \(\mu(a)=a\), \(q(a)=a^2\), and
\(\Phi(a)=a^2\).  Finally,
\(\E[(Z(b)-b)^2]=b^2-\Phi(b)\), so
\[
    \E[(Z(b)-b)^2]
    =b^2-a^2-2\int_a^b\mu(s)\,ds
    =2\int_a^b(s-\mu(s))\,ds. \qedhere
\]
\end{proof}

We next prove the lower bound for a balanced \(d\)-coefficient prediction
space.  This is the auxiliary instance used in the reduction to affine
regression.

\begin{theorem}[Balanced prediction-space lower bound]
\label{thm:lower-balanced}
Let \(n=d+1\), and let \(X\in\R^{n\times d}\) have full column rank with
\[
    \operatorname{col}(X)
    =L:=\left\{z\in\R^n:\sum_{i=1}^n z_i=0\right\}.
\]
Every randomized mechanism whose outputs lie in \(L\) and that is
strategyproof in expectation for quadratic individual loss has squared-loss
approximation ratio at least \(d+1\).
\end{theorem}

\begin{proof}
Suppose that such a mechanism has a finite approximation ratio \(c\).  For a
report vector \(y\in\R^n\), let \(Z(y)\in L\) be its random prediction vector,
and write \(\OPT(y):=\min_{z\in L}\norm{z-y}_2^2\).

If \(\sigma(y):=\sum_i y_i\), orthogonal projection onto \(L\) gives
\(\OPT(y)=\sigma(y)^2/n\).  In particular, if \(y\in L\), then finite
approximability forces
\(\E[\norm{Z(y)-y}_2^2]=0\), and therefore \(Z(y)=y\) almost surely.

For each agent \(i\), define \(C_i(y):=\E[(Z_i(y)-y_i)^2]\) and
\(\mu_i(y):=\E[Z_i(y)]\).

Fix \(\tau>0\) and a report vector \(y\) with \(\sigma(y)=\tau\).  If agent
\(i\) changes only their report from \(y_i\) to \(y_i-\tau\), the resulting
vector \(y-\tau e_i\) belongs to \(L\), so the mechanism is exact there.
Applying Lemma~\ref{lemma:envelope} while holding \(y_{-i}\) fixed, and then
putting \(s=r-(y_i-\tau)\), gives

\begin{equation}
\label{eq:lower-path-envelope}
    C_i(y)=2\int_0^\tau
    \left(y_i-\tau+s
    -\mu_i\bigl(y-(\tau-s)e_i\bigr)\right)ds.
\end{equation}

We average~\eqref{eq:lower-path-envelope} over large balls in \(L\).  For
\(R>0\), let \(K_R:=\{w\in L:\norm{w}_2\le R\}\),
and let \(W_R\) be uniform on \(K_R\) with respect to the \(d\)-dimensional
Lebesgue measure on \(L\).  Set
\(Y_R:=W_R+(\tau/n)\mathbf 1\); every realization of \(Y_R\) has coordinate
sum \(\tau\) and optimum \(\tau^2/n\).  For \(w\in L\) and
\(s\in[0,\tau]\), define
\[
    \Delta_i(w,s):=
    w_i+\frac{s}{n}
    -\mu_i\left(w+\frac{s}{n}\mathbf 1\right),
    \qquad
    p_i:=e_i-\frac1n\mathbf 1\in L.
\]
The profile on the path in~\eqref{eq:lower-path-envelope} decomposes as
\[
    Y_R-(\tau-s)e_i
    =\bigl(W_R-(\tau-s)p_i\bigr)+\frac{s}{n}\mathbf 1.
\]
Consequently, the integrand in~\eqref{eq:lower-path-envelope} is
\(\Delta_i(W_R-(\tau-s)p_i,s)\).

We now justify replacing the translated ball by the original ball.  At a
profile \(v=w+(s/n)\mathbf 1\), where \(w\in L\), we have
\(\OPT(v)=s^2/n\).  The approximation guarantee and Jensen's inequality give

\begin{equation}
\label{eq:lower-delta-bound}
    |\Delta_i(w,s)|
    \le \sqrt{\E[(Z_i(v)-v_i)^2]}
    \le \sqrt{\frac cn}\,s.
\end{equation}

Apply the ball-shift estimate \eqref{eq:absolute-ball-shift} to
\(f(w)=\Delta_i(w,s)\) with shift \(-(\tau-s)p_i\). By
\eqref{eq:lower-delta-bound} and \(\norm{p_i}_2\le1\), for \(R>\tau\)
and \(s\in[0,\tau]\),
\[
\left|
    \E[\Delta_i(W_R-(\tau-s)p_i,s)]
    -\E[\Delta_i(W_R,s)]
\right|
\le \kappa_n\sqrt{\frac cn}\,\frac{s(\tau-s)}{R}.
\]
By Fubini's theorem and \eqref{eq:lower-path-envelope},

\begin{equation}
\label{eq:lower-averaged-envelope}
    \E[C_i(Y_R)]
    =2\int_0^\tau\E[\Delta_i(W_R,s)]\,ds+\varepsilon_{i,R},
    \qquad
    |\varepsilon_{i,R}|
    \le \kappa_n\sqrt{\frac cn}\,\frac{\tau^3}{3R}.
\end{equation}

For fixed \(w\in L\) and \(s\in[0,\tau]\), let
\(v=w+(s/n)\mathbf 1\).  Every realization \(Z(v)\) belongs to \(L\), so
\(\sum_i\mu_i(v)=0\).  Since \(\sum_i v_i=s\), we therefore have the pointwise
formula
\[
    \sum_{i=1}^n\Delta_i(w,s)=s.
\]
Summing~\eqref{eq:lower-averaged-envelope} over the agents and letting
\(R\to\infty\) now yields
\[
    \lim_{R\to\infty}
    \E\left[\sum_{i=1}^n C_i(Y_R)\right]
    =2\int_0^\tau s\,ds=\tau^2.
\]
On the other hand, pointwise \(c\)-approximability and
\(\OPT(Y_R)=\tau^2/n\) imply
\[
    \E\left[\sum_{i=1}^n C_i(Y_R)\right]
    \le c\frac{\tau^2}{n}
\]
for every \(R\).  Taking \(R\to\infty\) gives \(c\ge n=d+1\).
\end{proof}

\subsection{Reduction to ordinary affine regression}
\label{app:affine-lower}
We finally transfer the balanced instance to ordinary affine regression. 

\begin{theorem}[Affine regression with an intercept]
\label{thm:lower-affine}
For every \(d\ge1\), there is an unweighted public design for affine regression
on \(\R^{d-1}\), with hypothesis class
\[
    f_\beta(x)=\beta_0+\sum_{j=1}^{d-1}\beta_jx_j,
    \qquad \beta\in\R^d,
\]
such that every randomized mechanism that is strategyproof in expectation for
quadratic individual loss has squared-loss approximation ratio at least \(d+1\).
\end{theorem}

\begin{proof}
Choose \(d\) affinely independent points
\(x_1,\ldots,x_d\in\R^{d-1}\), and let
\(\bar x:=d^{-1}\sum_{i=1}^d x_i\).  Write
\(a_i^{\mathsf T}:=(1,x_i^{\mathsf T})\), let \(A\in\R^{d\times d}\) have
rows \(a_i^{\mathsf T}\), and put
\(\bar a:=d^{-1}\sum_i a_i=(1,\bar x^{\mathsf T})^{\mathsf T}\).
Affine independence makes \(A\) invertible.
For \(d=1\), this construction is constant regression.

Consider first the auxiliary \((d+1)\)-agent homogeneous design
\[
    \widetilde X
    :=\begin{pmatrix}A\\-d\bar a^{\mathsf T}\end{pmatrix}
    =\begin{pmatrix}I_d\\-\mathbf 1^{\mathsf T}\end{pmatrix}A.
\]
It has \(d\) free coefficients and
\(\operatorname{col}(\widetilde X)=\mathbf 1^\perp\).  Hence
Theorem~\ref{thm:lower-balanced} applies to this design.

Now form an affine instance with one agent at each \(x_i\), \(i\in[d]\), and
\(d^2\) additional agents at \(\bar x\).  Suppose, toward a contradiction, that
there is a mechanism \(M\), strategyproof in expectation, with approximation
ratio \(c<d+1\) for this fixed affine design.  We use \(M\) to construct a
mechanism \(\widetilde M\) for the auxiliary design.

Given auxiliary reports \(y=(y_1,\ldots,y_{d+1})\), give the affine agents at
\(x_i\) the reports \(y_i\), and give every one of the \(d^2\) agents at
\(\bar x\) the common report \(u:=-y_{d+1}/d\).

Run \(M\) on this affine report profile and let \(\widetilde M\) return the same
coefficient vector \(\beta\in\R^d\).  For every deterministic value of
\(\beta\), the auxiliary and affine sums of squared residuals agree:

\begin{align*}
    \norm{\widetilde X\beta-y}_2^2
    &=\sum_{i=1}^d(a_i^{\mathsf T}\beta-y_i)^2
      +(-d\bar a^{\mathsf T}\beta-y_{d+1})^2\\
    &=\sum_{i=1}^d(f_\beta(x_i)-y_i)^2
      +d^2\left(f_\beta(\bar x)+\frac{y_{d+1}}d\right)^2.
\end{align*}

The second line is exactly the affine social cost after the transformation of
reports.  Minimizing over \(\beta\) shows that the optima also agree.  Therefore
\(\widetilde M\) inherits the approximation ratio \(c\).

It remains to verify strategyproofness in expectation.  For any of the first
\(d\) auxiliary agents, a unilateral change of report changes only the report
of the corresponding affine agent, and both agents have the same prediction
and quadratic loss.  Strategyproofness of \(M\) therefore gives the desired
incentive inequality.

For auxiliary agent \(d+1\), fix a true label \(t\) and an alternative report
\(t'\), and set \(u:=-t/d\) and \(u':=-t'/d\).  Enumerate the \(m:=d^2\)
co-located agents.  For \(j=0,\ldots,m\), let \(\rho^j\) be the affine report
profile in which the first \(j\) co-located agents report \(u'\), the remaining
\(m-j\) report \(u\), and all other reports are fixed.  Let \(P_j\) denote the
random prediction at \(\bar x\) when \(M\) is run on \(\rho^j\).

Applying strategyproofness to agent \(j\), with true label \(u\) and all other
reports fixed, gives
\[
    \E[(P_{j-1}-u)^2]\le\E[(P_j-u)^2].
\]
All co-located agents receive the same prediction, so these inequalities chain
from \(j=1\) to \(m\).  Multiplying the resulting inequality by \(d^2\) and
using
\[
    (-dP_j-t)^2=d^2(P_j-u)^2
\]
shows precisely that auxiliary agent \(d+1\) cannot lower their expected
quadratic loss by changing their report from \(t\) to \(t'\).  Thus
\(\widetilde M\) is strategyproof in expectation.  Theorem~\ref{thm:lower-balanced}
then gives \(c\ge d+1\), a contradiction.
\end{proof}

\section{Randomization for general power losses}
\label{app:powers}

We prove \Cref{thm:powers} for the complete hypothesis class \(\Theta=\R^d\).
The proof uses two independent full-size volume-sampled interpolants. We first
establish the moments in \eqref{eq:two-fit-moments}, then verify the incentives and squared-error
guarantees for the two lotteries in \Cref{sec:powers}.

\subsection{Moments of the two interpolating fits}

For a set \(B\subseteq[n]\) of size \(d\), let \(X_B\) be the corresponding
square row submatrix.  Full-size volume sampling draws \(B\) according to
\begin{equation*}
  \Prb(B)=\frac{\det(X_B)^2}{\det(X^\top X)}.
  \label{powers:eq-volume-law}
\end{equation*}
The probabilities sum to one by Cauchy--Binet.  Singular sets have probability
zero; on every supported set, define the interpolating prediction vector
\(U(y)=XX_B^{-1}y_B\).  Define \(V(y)\) from an independent set \(C\) drawn by
the same law. Since $\det(H_B)=\det(X_B)^2/\det(X^\top X)$ when $|B|=d$, this is exactly \eqref{eq:partial-interpolation} with $b=d$. The interpolation constraints then determine the fit uniquely, so $U=Z^B(y)$.

The following lemma is due to \citet{derezinski2017}. We give a direct proof.

\begin{lemma}[Moments of a volume-sampled interpolant]
\label[lemma]{powers:volume-moments}
Let \(z^\star=Hy\) and \(r=(I-H)y\).  Then
\[
  \E U=\E V=z^\star,
  \qquad
  \sigma^2:=\E\norm{U-z^\star}_2^2
  =\E\norm{V-z^\star}_2^2
  \le d\norm r_2^2=d\OPT(y).
\]
\end{lemma}

\begin{proof}
For each sampled set \(B\), put \(u_B=X_B^{-1}r_B\), so \(U-Hy=Xu_B\).
For each \(j\in[d]\), let \(Y^{(j)}\) be the matrix obtained from \(X\) by
replacing column \(j\) by \(r\).  Cramer's rule gives
\(\det(X_B)(u_B)_j=\det(Y^{(j)}_B)\). Cauchy--Binet therefore gives
\begin{align*}
  \E (u_B)_j
  &=\frac{1}{\det(X^\top X)}
    \sum_{|B|=d}\det(X_B)\det(Y^{(j)}_B)\\
  &=\frac{\det(X^\top Y^{(j)})}{\det(X^\top X)}=0.
  \label{powers:eq-unbiased-cramer}
\end{align*}
Singular sets contribute zero to the sum, and the final determinant vanishes
because its \(j\)-th column is \(X^\top r=0\). Hence \(\E U=Hy\).

For the second moment, $U-Hy\in L$ and $r=y-Hy\perp L$, so
\[
  \E\norm{U-y}_2^2
  =\norm r_2^2+\E\norm{U-Hy}_2^2.
\]
Applying \Cref{lem:partial-efficiency} with $b=d$ bounds the left-hand side by
$(d+1)\norm r_2^2$, giving
$\E\norm{U-Hy}_2^2\le d\norm r_2^2$. The independent copy $V$ has the same moments.
\end{proof}

\subsection{The convex range \texorpdfstring{$p\ge1$}{p >= 1}}

Let \(\kappa=1+\sqrt2\), and let \(\lambda_p\) be as in
\Cref{thm:powers}.  Put
\[
  a=\frac{\sqrt2-1}{2}=\frac{1}{2\kappa},
  \qquad
  b=\frac{\sqrt2+1}{2}=\frac{\kappa}{2}.
\]
The three-outcome lottery in \eqref{eq:power-lottery} can be written
\(Z=(1-T)U+TV\), where \(T=1/2\) with probability \(1-\lambda_p\), and
\(T=-a\) and \(T=b\) each with probability \(\lambda_p/2\).
Draw \(T\) independently of the sampled sets and reports. Notice that
\(b=1+a\), so \(T\) and \(1-T\) have the same distribution.

For \(x\ne0\), define
\(\psi_p(x)=\operatorname{sgn}(x)|x|^{p-1}\).  Direct substitution gives the
incentive balance
\begin{align}
 \E[(1-T)\psi_p(T)]
 &= (1-\lambda_p)2^{-p}
    -\frac{\lambda_p}{2}
       \bigl(ba^{p-1}+ab^{p-1}\bigr)=0,
 \label{powers:eq-convex-balance}
\end{align}
because
\[
  2^{p-1}\bigl(ba^{p-1}+ab^{p-1}\bigr)
  =\frac{\kappa^{p-2}+\kappa^{2-p}}{2}
  =\frac{1-\lambda_p}{\lambda_p}.
\]

\begin{lemma}[Incentives for the convex-power lottery]
\label[lemma]{powers:convex-incentives}
For every \(p\ge1\), the corresponding lottery in
\eqref{eq:power-lottery} is SPIE for individual loss \(|z_i-y_i|^p\).
\end{lemma}

\begin{proof}
Condition on the sampled sets \(B,C\), fix agent \(i\), and fix all other
reports.  If \(i\notin B\cup C\), their report affects neither fit.  If
\(i\in B\cap C\), both fits interpolate their report, and hence every lottery
outcome does too; truthful reporting gives zero loss.

Suppose \(i\in B\setminus C\).  Let \(t\) be their true label and let \(v\) be
the prediction at \(i\) of the fit based on \(C\), which is independent of their
report. If they report \(t+\delta\), their conditional expected loss is
\[
  F(\delta)=\E_T
  \left|T(v-t)+(1-T)\delta\right|^p.
\]
This is a convex function of \(\delta\). If \(v=t\), then \(F(0)=0\).
Otherwise, \(T(v-t)\ne0\) for every outcome, so \(F\) is differentiable at
zero even when \(p=1\), and
\[
  F'(0)=p\,\psi_p(v-t)\,
  \E[(1-T)\psi_p(T)]=0
\]
by \eqref{powers:eq-convex-balance}.  Convexity makes \(\delta=0\) a global
minimizer.  The case \(i\in C\setminus B\) follows from the equality in law of
\(T\) and \(1-T\).  The incentive inequality therefore holds for every fixed
pair \((B,C)\), and averaging over their report-independent distribution
proves SPIE.
\end{proof}

\begin{lemma}[Squared error of the convex-power lottery]
\label[lemma]{powers:convex-efficiency}
For \(p\ge1\), the expected squared error is at most
\([1+d(1/2+\lambda_p)]\OPT(y)\).
\end{lemma}

\begin{proof}
Let \(z^\star=Hy\), \(\xi=U-z^\star\), and \(\eta=V-z^\star\).  By
\Cref{powers:volume-moments}, \(\xi\) and \(\eta\) are independent, mean-zero,
and have the same second moment \(\sigma^2\).  They lie in \(L\), while
\(y-z^\star\perp L\).  Since \(T\) is independent of both fits,
\begin{align}
  \E\norm{Z-y}_2^2
  &=\OPT(y)+\E\norm{(1-T)\xi+T\eta}_2^2 \notag\\
  &=\OPT(y)+\E\bigl[(1-T)^2+T^2\bigr]\sigma^2.
  \label{powers:eq-variance-decomposition}
\end{align}
The cross term vanishes because
\(\E\langle\xi,\eta\rangle=\langle\E\xi,\E\eta\rangle=0\).  At \(T=1/2\),
the multiplier is \(1/2\); at either \(T=-a\) or \(T=b\), it is
\(a^2+b^2=3/2\).  Consequently
\[
  \E\bigl[(1-T)^2+T^2\bigr]
  =(1-\lambda_p)\frac12+\lambda_p\frac32
  =\frac12+\lambda_p.
\]
The result follows from \(\sigma^2\le d\OPT(y)\).
\end{proof}

\subsection{Powers below one: \texorpdfstring{$0<p<1$}{0 < p < 1}}

Set \(w_p=(2+2^p)^{-1}\). Independently of the sampled sets and reports,
return \(U\) and \(V\) with probability \(w_p\) each, and their midpoint
with probability \(2^p w_p\).

\begin{lemma}[Sublinear-power incentives and efficiency]
\label[lemma]{powers:concave-case}
For \(0<p<1\), this lottery is SPIE for individual loss
\(|z_i-y_i|^p\), and its expected squared error is at most
\[
  \left[1+d\frac{2+2^{p-1}}{2+2^p}\right]\OPT(y).
\]
\end{lemma}

\begin{proof}
Condition again on \((B,C)\).  The cases where agent \(i\) belongs to neither
set or to both sets are immediate.  If \(i\in B\setminus C\), use the notation from
the proof of \Cref{powers:convex-incentives} and set \(s=v-t\).  After a report
\(t+\delta\), their expected loss is
\[
  F(\delta)=w_p\bigl(
     |\delta|^p+|s|^p+|s+\delta|^p\bigr).
\]
For \(0<p<1\), the triangle inequality followed by subadditivity of
\(x\mapsto x^p\) on \(\R_+\) gives
\[
  |s|^p\le |s+\delta|^p+|\delta|^p.
\]
Hence \(F(\delta)\ge2w_p|s|^p=F(0)\).  The case
\(i\in C\setminus B\) is symmetric, so the mechanism is SPIE.

For the squared-error calculation, write the lottery as
\(Z=(1-T)U+TV\), with \(T=0,1,1/2\) having probabilities
\(w_p,w_p,2^p w_p\), respectively.  The variance decomposition
\eqref{powers:eq-variance-decomposition} remains valid, and
\[
  \E[(1-T)^2+T^2]
  =2w_p+2^{p-1}w_p
  =\frac{2+2^{p-1}}{2+2^p}.
\]
Apply \Cref{powers:volume-moments} to finish the proof.
\end{proof}

\begin{proof}[Proof of \Cref{thm:powers}]
For \(p\ge1\), combine \Cref{powers:convex-incentives} and
\Cref{powers:convex-efficiency}; for \(0<p<1\), use
\Cref{powers:concave-case}.  These give exactly the factor
\(1+dq_p\) in \eqref{eq:power-factor}.

Each outcome is a linear combination of \(U\) and \(V\), so it belongs to
\(L=\col(X)\) and is feasible for the complete hypothesis class
\(\Theta=\R^d\).

Finally, when \(0<p<1\), the numerator defining \(q_p\) is strictly smaller
than its denominator.  When \(p\ge1\), the arithmetic--geometric mean
inequality gives
\(\kappa^{p-2}+\kappa^{2-p}\ge2\), with equality exactly at \(p=2\).
Therefore \(\lambda_p\le1/2\), again with equality exactly at \(p=2\), so
\(q_2=1\) and \(q_p<1\) for every \(p\ne2\).
\end{proof}

\paragraph{Behavior near quadratic loss.}
For our lottery with \(p\ge1\), putting \(x=\kappa^{p-2}\) gives
\[
  \lambda_p=\frac{2x}{(x+1)^2},
  \qquad
  1+dq_p
  =d+1-\frac d2\left(\frac{x-1}{x+1}\right)^2.
\]
The improvement in this bound vanishes quadratically as \(p\to2\):
\begin{equation*}
  \frac d2\left(
    \frac{\kappa^{p-2}-1}{\kappa^{p-2}+1}
  \right)^2
  =\frac{d\log^2\!\kappa}{8}(p-2)^2
   +O\!\left(d(p-2)^4\right).
  \label{powers:eq-near-two}
\end{equation*}

\end{document}